\documentclass[letterpaper, 10pt, conference]{ieeeconf}      
\IEEEoverridecommandlockouts                              
\usepackage{graphics}
\usepackage{amssymb}  
\usepackage{mathtools}

\usepackage{hyperref}
\usepackage{cite}
\usepackage{multicol}
\usepackage{color}
\usepackage{url}
\usepackage{caption}
\usepackage{subcaption}

\newtheorem{theorem}{Theorem}

\newtheorem{definition}{Definition}
\newtheorem{proposition}{Proposition}

\title{\LARGE \bf
	Analog cross coupled controller for oscillations: \\ modeling and design via dominant system theory}

\author{Weiming Che, Thomas Chaffey, Fulvio Forni
	\thanks{W. Che is supported by CSC Cambridge Scholarship. T. Chaffey is partially supported by the Advanced ERC Grant Agreement Switchlet n. 670645. W. Che, T. Chaffey and F. Forni are with the Department of Engineering, University of Cambridge, CB2 1PZ, UK {\tt\small wc289|tlc37|f.forni@eng.cam.ac.uk}}}

\begin{document}
	
	\maketitle
	\thispagestyle{empty}
	\pagestyle{empty}

	\begin{abstract}
	 We propose a new analog feedback controller based on the classical cross coupled electronic oscillator. The goal is to drive a linear passive plant into oscillations. We model the circuit as Lur'e system and we derive a new graphical condition to certify oscillations (inverse circle criterion for dominance theory). These conditions are then specialized to minimal control architectures like RLC and RC networks, and are illustrated with an example
	 based on a DC motor model.
	\end{abstract}

	\section{Introduction}
    The cross coupled oscillator is a classical circuit architecture in RF technology, widely employed in electronic devices to realize function generators, phase-lock loops, frequency synthesizers, etc.  \cite[Chapter 8]{razavi2012rf}. The basic configuration of the cross coupled oscillator contains a cross coupled pair, XCP, and two identical RLC tank circuits, $C$. This configuration corresponds to the circuit in Figure \ref{fig:XCP_plant}, when the additional circuit $P$ is removed (open circuit). The role of the XCP as a bounded differential negative resistance is to restore the energy lost over time by the resonant RLC tank circuits, to achieve steady oscillations. From a system-theoretic perspective 
    the XCP introduces positive feedback to the circuit \cite{razavi2014cross,razavi2015cross}, which is crucial for the
    the realization of a wide range of robust non-equilibrium behaviors, like oscillators \cite{hu1986self,li2000active}, bistable switches \cite{kennedy1991hysteresis,chen2009negative} and chaos \cite{kennedy1993three}. 

	In this paper, we adapt the classic cross coupled oscillator for control purposes. The goal is to develop an analog controller that drives a plant into oscillations, in closed loop. In our approach, $C$ is a generic linear control circuit that we will design. The circuit will control a voltage driven linear passive plant, $P$, through physical port interconnection, as shown in Figure \ref{fig:XCP_plant}. The controller can potentially drive any voltage driven electro-mechanical devices, such as electric motors, solenoids, or electroactive polymers. 
	The main motivation for our study comes from robotic locomotion. Walking, running, and swimming are all characterized by specific rhythmic patterns that regulate the interaction with the environment
	\cite{kimura1999realization},\cite{ijspeert2007swimming,ijspeert2020amphibious}.  In these settings, microcontrollers are typically used to generate rhythmic reference signals  to which the robot body is entrained via an amplification / actuation stage, see e.g. \cite[Table 1]{ijspeert2020amphibious}. This means that the mechanical features of the robot body and its interaction with the environment do not affect the frequency of the rhythm. In contrast, our controller is an analog circuit that directly drives the plant $P$, generating oscillations in feedback from the plant. This means that the plant dynamics play an active role in the generation of oscillations. Frequency and shape of the generated rhythm are thus sensitive to the features of the robot/environment,
	potentially enabling their adaptation in feedback from
	the environment.
	
    In electronics, the design of the cross coupled oscillator is usually approached via local methods, like  the Barkhausen’s criterion \cite{razavi2012rf}. The idea is to use positive feedback to destabilize the equilibrium of the RLC tank circuit in a controlled way, to achieve oscillations. Oscillations are certified by the Poincar{\'e}-Bendixson theorem \cite{Hirsch1974}, taking advantage of the planar dynamics of RLC tank circuits. This makes it difficult to scale these approaches to the case of large dimensional control networks $C$ and plants $P$. Instead, 
    our approach is based on the theory of dominant systems, which provides various tools for the analysis and design of robust non-equilibrium behaviors \cite{forni2018differential,miranda2018analysis}, and has been used to analyze oscillatory circuits based on mixed feedback \cite{che2021dominant} and negative resistors \cite{miranda2022dissipativity}. The intuition is that the attractors of a  $2$-dominant system correspond to the attractors of a planar system. Thus, the Poincar{\'e}-Bendixson theorem can be used on a large dimensional $2$-dominant system to certify oscillations. Indeed, our design will achieve oscillations by building a closed-loop system that is $2$-dominant, has only unstable equilibria, and has bounded trajectories. 
    
\begin{figure}[h]
\centering
\begin{subfigure}{.24\textwidth}
  \centering
  \includegraphics[width=.8\linewidth]{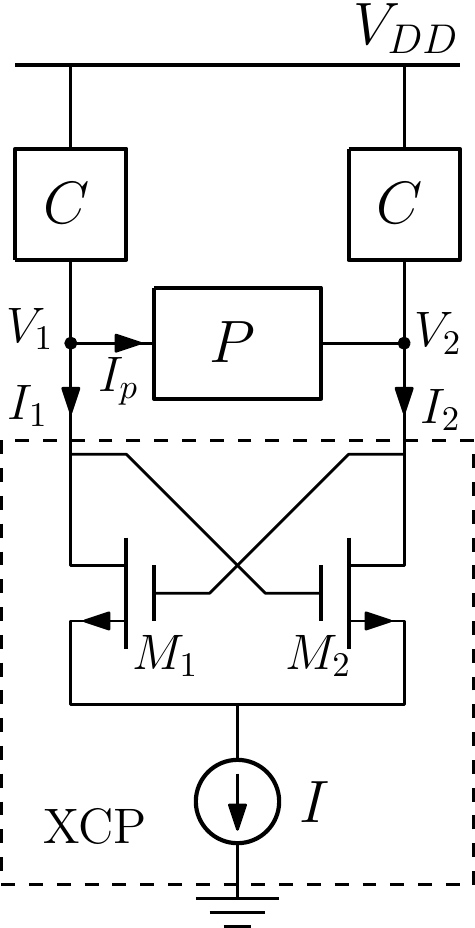}
  \caption{Circuit representation.}
  \label{fig:XCP_plant}
\end{subfigure}%
\begin{subfigure}{.24\textwidth}
  \centering
  \vspace{9mm}
  \includegraphics[width=.8\linewidth]{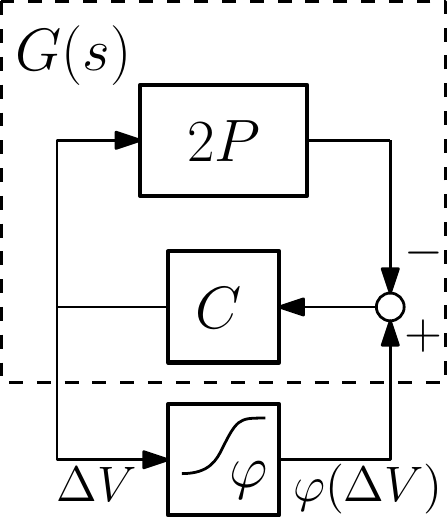}
  \vspace*{18mm}
  \caption{Block diagram representation.}
  \label{fig:XCP_OSC_Block}
\end{subfigure}
\caption{The cross coupled oscillator circuit driving a voltage controlled linear passive plant $P$.}
\end{figure}
    
    We will model the cross coupled controller as a Lur'e feedback system, as shown in Figure \ref{fig:XCP_OSC_Block}. The port connection of the plant network $P$ and controller network $C$ forms the network $G$ which is in positive feedback with a sigmoidal nonlinearity $\varphi$. The latter 
    physically implemented by XCP. Our goal is to design $C$ and the slope of $\varphi$ to achieve oscillations. Our design is based on a new inverse circle criterion for dominance, combined with feedback gain tuning for the instability of closed-loop equilibria. 
    The paper is organized as follows. Section \ref{sec:modeling} introduces the cross coupled controller and  derives the block diagram in Figure \ref{fig:XCP_OSC_Block}. Section \ref{sec:dominance} gives a brief introduction to dominance theory and develops graphical conditions / the inverse circle criterion to achieve $2$-dominance in closed loop. Section \ref{sec:instability} derives conditions on the linear control circuit $C$ and on the nonlinearity $\varphi$ for unstable closed-loop equilibria. These conditions combined to the results on $2$-dominance provide a complete characterization of the design procedure for oscillations.
    Section \ref{sec:RLCcontroller} shows how to achieve oscillation with control circuits $C$ limited to simple $RLC$ and $RC$ networks. 
    Section \ref{sec:example} completes the paper with a design example of a DC motor controlled into oscillation. Conclusion follows.
    
\section{The cross coupled controller}\label{sec:modeling}
We extend the use of  cross coupled oscillator as an analog controller to drive the oscillation of a plant $P$, as shown in Figure \ref{fig:XCP_plant}. The plant $P$ is an admittance driven by the voltage difference, $V_1-V_2$, which determines the  output current  $I_p$. There is no other power supply therefore $P$ is a passive circuit. We assume $P$ is represented by the transfer function $P(s)$. The controller $C$ is also a linear circuit, whose impedance is represented by the  transfer function $C(s)$. 

The transistors $M_1$, $M_2$ and the current source form the XCP circuit and provide the main nonlinearity of the system. By taking the large signal description of the transistors, such nonlinearity is characterized by a sigmoid function $\varphi$. This leads to the equivalent block diagram for the cross coupled controller  shown in Figure \ref{fig:XCP_OSC_Block}.

The derivation of the nonlinearity $\varphi$ is based on the Shichman-Hodges equation \cite{shichman1968modeling} and on the large signal description of the MOSFET differential amplifier, e.g. see \cite[Chapter 9.1]{sedra2004microelectronic}. Taking $I_1$ and $I_2$ as the currents that flow through the transistors $M_1$ and $M_2$ respectively, we have
\begin{equation*}
	I_{i}=\begin{cases}
		\frac{I}{2} \!+\! (-1)^{i-1}\frac{\sqrt{k_nI}\Delta V}{2}\sqrt{1-\frac{k_n\Delta V^2}{2I}}, \ &|\Delta V|\leq\sqrt{\frac{2I}{k_n}}\\
		\frac{I}{2}\!+\! (-1)^{i-1}\frac{I}{2}\text{sgn}(\Delta V), \ &\text{Otherwise}
	\end{cases}
\end{equation*}
where 
$\Delta V := V_2-V_1 = V_{G_1}-V_{G_2}$ is the differential input voltage to the gates of $M_1$ and $M_2$. $k_n$ is the transistor gain and $I$ is the total current flowing in the circuit, modulated by the current source in Figure \ref{fig:XCP_plant}.

By taking $\Delta I := I_1-I_2$, the nonlinear differential relationship $\Delta I=\varphi(\Delta V)$ reads
\begin{equation}\label{eq:XCP_IV_relationship}
	\varphi(\Delta V)=\begin{cases}
		\sqrt{k_nI}\Delta V\sqrt{1-\frac{k_n\Delta V^2}{4I}},\quad &|\Delta V|\leq\sqrt{\frac{2I}{k_n}}\\
		I\text{sgn}(\Delta V), \quad &\text{Otherwise}
	\end{cases}
\end{equation}
which corresponds to a differentiable sigmoid function with the sector condition
\begin{equation}\label{eq:sector_cond}
    \partial\varphi\in[0,K], \quad K=\sqrt{k_nI}.
\end{equation} The slope of $\varphi$ in the linear regime near zero is controlled by the total current $I$.

For the linear components of the circuit, 
\begin{equation*}
    I_p=P(s)(V_1-V_2) = -P(s)\Delta V.
\end{equation*}
Thus, using Kirchhoff circuit law,
\begin{equation*}
	\begin{cases}
	V_1=V_{DD}-(I_1+I_p)C(s),\\
	V_2=V_{DD}-(I_2-I_p)C(s),
	\end{cases}
\end{equation*}
which leads to
\begin{equation}\label{eq:G(s)}
	\Delta V=\frac{C(s)}{1+2P(s)C(s)}\Delta I =: G(s) \Delta I.
\end{equation}
G(s) corresponds to the passive transfer function of the negative feedback of $C(s)$ with $2P(s)$.

\eqref{eq:XCP_IV_relationship} and \eqref{eq:G(s)} 
together form the closed-loop system  represented in Figure \ref{fig:XCP_OSC_Block}. The closed loop given by the XCP, the controller $C$ and the plant $P$ is a Lur'e feedback system, represented by the  transfer function $G(s)$ in \emph{positive feedback} with the $\varphi$. Our \emph{control objective} is to drive the closed loop into steady oscillation by (i)~shaping $G(s)$ via the design of $C(s)$ and by (ii)~tuning the nonlinear positive feedback $\varphi$ via the current $I$. 

\section{Control design for 2-dominance} \label{sec:dominance}
\subsection{Dominant systems} \label{sec:dom_sys}
We will use dominance theory \cite{forni2018differential} to derive oscillation conditions for the XCP controller. We recall here a few basic definitions and results on dominance theory that we will need in the rest of the paper. 

\begin{definition}
		The nonlinear system $\dot{x}=f(x)$, $x\in\mathbb{R}^n$, is $p$-\emph{dominant with rate} $\lambda\geq0$ if and only if there exist a symmetric matrix $P$ with inertia $(p,0,n-p)$ and $\varepsilon\geq0$ such that the matrix inequality
		\begin{equation}
		\label{eq:dominance_LMI}
			\partial f(x)^TP + P\partial f(x) + 2\lambda P  \leq -\varepsilon I
		\end{equation}
		holds for all $x\in\mathbb{R}^n$. The property is strict if $\varepsilon>0$.$\hfill\lrcorner$
	\end{definition}
	The inertia $(p,0,n-p)$ means that $P$ has $p$ negative real eigenvalues, $0$ imaginary eigenvalues, and $n-p$ positive real eigenvalues. Feasibility of \eqref{eq:dominance_LMI} indicates that the system consists of slow dominant dynamics of dimension $p$, and fast transient dynamics of dimension $n-p$, split by the rate $\lambda$. The asymptotic behavior of a $p$-dominant system
	corresponds to the behavior of a system of dimension $p$. This claim is made precise for $2$-dominance by the following theorem, from \cite{forni2018differential}. 
\begin{theorem}\label{th:p-attractor}
		For a strict $2$-dominant system with rate $\lambda \ge 0$, every bounded trajectory asymptotically converges to a simple attractor, that is, a fixed point, a set of fixed points and connecting arcs, or a limit cycle (oscillation).$\hfill\lrcorner$
	\end{theorem}

The asymptotic behavior of a $2$-dominant system corresponds to the behavior of a planar system. This allows us to use  Poincar{\'e}-Bendixson theory on $2$-dominant systems, even if their state space is $\mathbb{R}^n$. This suggests a route for the design of the cross coupled controller to achieve oscillations in closed loop. First, we will shape $G(s)$ in \eqref{eq:G(s)} and the nonlinearity $\varphi$ in \eqref{eq:XCP_IV_relationship} such that the closed-loop system is $2$-dominant. Then, using Poincar{\'e}-Bendixson, we will
ensure that every equilibria of the system are unstable. As a result, from Theorem \ref{th:p-attractor}, every bounded trajectory of the system will converge to a limit cycle, achieving oscillations. Note that boundedness of trajectories is always guaranteed, as long as $G(s)$ is a stable transfer function (by BIBO stability).

\subsection{Dominant cross coupled controller} \label{sec:XCP_dominance}
For Lur'e systems, dominance can be verified using graphical conditions like the circle criterion for dominance  \cite[Corollary 4.5]{miranda2018analysis}. Given the sector condition \eqref{eq:sector_cond}, the closed loop in Figure \ref{fig:XCP_OSC_Block} is $2$-dominant with rate $\lambda \geq 0$ if $G(s-\lambda)$ has two unstable poles, no poles on the imaginary axis, and if its Nyquist plot remains to the left of the region of the complex plane
$ \{z \in \mathbb{C}\,|\, \Re(z) \geq \frac{1}{K}\}$.  However, we can not directly use this condition for design. The linear subsystem represented by $G(s)$ in \eqref{eq:G(s)} corresponds to the negative feedback of $2P(s)$ and $C(s)$, which involves multiplication and inversion of transfer functions. Thus it is unclear how do we design $C(s)$ to achieve an admissible Nyquist plot of $G(s-\lambda)$. 

To address such issue, we consider the inverse of the closed loop. The idea is to work with
inverse elements $\varphi^{-1}$ (relational inverse) and $G^{-1}(s)$ to derive an adapted circle criterion 
that allows for the design of $C(s)$. 
From the sector condition \eqref{eq:sector_cond}, we have that $\varphi^{-1}$ satisfies 
\begin{equation}
    \label{eq:inverse_sector}
    \partial\varphi^{-1}\in[1/K,\infty].
\end{equation}
The inverse of $G(s)$ reads 
\begin{equation}\label{eq:G_inv}
	G^{-1}(s)=C^{-1}(s)+2P(s),
\end{equation}
which expresses the feedback relation between $C(s)$ and $P(s)$ as the addition of two transfer functions, one of which inverted. This inverse representation of the cross coupled controller preserves all information of the closed loop but provides an insightful perspective for control design.

\begin{theorem}[Inverse circle criteria for $2$-dominance]\label{th:inverse_circle_cirteria}
		Consider the Lur'e feedback system in Figure \ref{fig:XCP_OSC_Block}, with the static nonlinearity $\varphi$ satisfying sector condition \eqref{eq:sector_cond}. Then the closed system is strictly $2$-dominant with rate $\lambda$ if:
		\begin{enumerate}
			\item $C^{-1}(s)+2P(s)$ has no zeros with real part equal to $-\lambda$;
			\item the Nyquist plot of $C^{-1}(s-\lambda)+2P(s-\lambda)$ has $2-(q+r)$ clockwise encirclement of the origin, where
		\begin{equation*}
		\begin{split}
		    q&= \text{number of poles of $P(s-\lambda)$ in }\mathbb{C}^+,\\
		   r&=\text{number of zeros of $C(s-\lambda)$ in }\mathbb{C}^+.
		\end{split}
		\end{equation*}
			\item the Nyquist plot of $C^{-1}(s-\lambda)+2P(s-\lambda)$ lies outside the disk
			\begin{equation}\label{eq:disk}
				\{z\in\mathbb{C}\,|\, |z-K/2| \leq K/2\}. \vspace{-2mm}
			\end{equation}$\hfill\lrcorner$
		\end{enumerate}
	\end{theorem}
	\begin{proof}
     $1)$ directly follows from condition (ii) in \cite[Corollary 4.5]{miranda2018analysis}, since $C^{-1}(s)+2P(s)$ is the inverse of $G(s)$, whose zeros are the poles of $G(s)$.
	
	$2)$ corresponds to condition (iii)\cite[Corollary 4.5]{miranda2018analysis}, which requires that $G(s)$ has $2$ dominant poles. Representing controller and the plant  transfer functions as ratio of polynomials
	\[C(s)=\frac{n_c(s)}{d_c(s)},\quad P(s)=\frac{n_p(s)}{d_p(s)}\]
	we have
	\[\begin{split}
	    C^{-1}&(s-\lambda)+2P(s-\lambda) =\\
	    &=\frac{2n_c(s-\lambda)n_p(s-\lambda)+d_c(s-\lambda)d_p(s-\lambda)}{n_c(s-\lambda)d_p(s-\lambda)}.
	\end{split}\]
	The roots of the numerator are the poles of $G(s-\lambda)$. The roots of the denominator are the zeros of $C(s-\lambda)$ and the poles of $P(s-\lambda)$. Following the principle of the argument in complex analysis, $G(s)$ has $2$ dominant poles if and only if the Nyquist plot of $C^{-1}(s-\lambda)+2P(s-\lambda)$ has $2-(q+r)$ clockwise encirclement of the origin.
	
	$3)$ adapts condition (iv.c) of \cite[Corollary 4.5]{miranda2018analysis} to the positive feedback between $G^{-1}(s)$ and $\varphi^{-1}$. First, consider $G(s)$ in negative feedback with $\Tilde{\varphi}=-\varphi$, whose sector conditions reads
	$\partial\Tilde{\varphi}\in[-K,0]$. Then, from condition (iv.c) of  \cite[Corollary 4.5]{miranda2018analysis}, the Nyquist plot of $G(s-\lambda)$ should remain to the left of the vertical line passing through the point  $1/K$. This separation must be preserved when we take the inverse of both $G(s)$ and $\Tilde{\varphi}$. That is, the region to the right of the vertical line passing through $1/K$ maps into the interior of the disk \eqref{eq:disk} and the Nyquist plot of $G^{-1}(s-\lambda)$ must remain outside the disk \eqref{eq:disk}.
	\end{proof}

Theorem \ref{th:inverse_circle_cirteria} further decouples the linear part of the Lur'e system
in Figure \ref{fig:XCP_OSC_Block} into the addition of the plant $P$ and of the inverse of the controller $C$. This simplifies
the design of the controller transfer function $C(s)$.  Given the plant $P(s)$, one can simply derive the Nyquist plot of the shifted transfer function $2P(s-\lambda)$ and the disk \eqref{eq:disk} on the complex plane. Then, one can design the inverse transfer function $C^{-1}(s-\lambda)$ to shape this Nyquist plot, to guarantee that  
$C^{-1}(s-\lambda)+2P(s-\lambda)$ fulfills the requirements of Theorem \ref{th:inverse_circle_cirteria}. This is discussed more in detail in Section \ref{sec:RLCcontroller}.

\section{Controlled instability in $2$-dominant systems for oscillations}\label{sec:instability}
As discussed in the previous section, our design of a closed loop oscillator requires
a $2$-dominant closed loop whose equilibria are all \emph{unstable} (Theorem \ref{th:p-attractor}). For simplicity, we will restrict our design to the case where the closed-loop system has only one equilibrium point, at the origin. For instance, 
the equilibrium points of the cross-coupled circuit in Figure \ref{fig:XCP_OSC_Block} 
must satisfy
\begin{equation}\label{eq:det_fixed_point}
    \frac{1}{G(0)}\Delta V=\varphi(\Delta V),
\end{equation}
where $G(0)\geq0$, due to passivity of $G(s)$. 
The case of a single equilibrium occurs when 
\begin{equation}
  \frac{1}{G(0)}\geq K ,
\end{equation}
where $K$ is the largest slope of the sigmoidal nonlinearity, given by $K = \sqrt{k_nI}$ as shown in \eqref{eq:sector_cond}. Instability is guaranteed for large
$K$, that is, for large current $I$.

\begin{proposition}\label{prop:unstable}
For a general stable linear system $G(s)$, the positive feedback system formed by $G(s)$ with gain $K$ will be unstable when $K$ is sufficiently large.
\end{proposition}

\begin{proof}
From root locus analysis: ``all sections of the real axis with an even number of poles and zeros to their right belong to the root locus for positive feedback''. Thus the right half real axis belongs to the root locus and instability is expected when the positive feedback gain $K$ is large. 
\end{proof}

Proposition \ref{prop:unstable} guarantees the instability of the origin when $K$ is large
but this is in conflict with the bound $K\leq\frac{1}{G(0)}$, which is needed 
to guarantee a single equilibrium.
This can be resolved by having a slow zero $z_s$ in $G(s)$ in addition to $2$-dominance. 
In fact, when $z_s$ gets close to the origin,
\[\lim_{z_s\rightarrow0}\frac{1}{G(0)}=\infty,\] 
i.e. $K$ can be arbitrarily large. The observation in Proposition \ref{prop:change_stable} below further clarifies the importance of a slow zero for our design.

\begin{proposition}\label{prop:change_stable}
    The positive root locus of a linear system $G(s)$ passes through the origin for the  gain $K=\frac{1}{G(0)}$.
\end{proposition}
\begin{proof}
The positive feedback root locus satisfies
\[1-KG(s)=0\]
For $K=\frac{1}{G(0)}$, $s=0$ is a pole of the feedback system.
\end{proof}

Proposition \ref{prop:change_stable} states that one pole of the positive feedback system changes its stability at $K=\frac{1}{G(0)}$. If the pole crosses the imaginary axis from the right half plane, then there must be $K<\frac{1}{G(0)}$ such that the origin is unstable. To ensure such transition, it is necessary for $G(s)$ to have a slow zero $z_s$ on the real axis, as illustrated by the positive feedback root locus of $G(s)$  in Figure \ref{fig:Rlocus_Zero}. For conciseness, we ignore the root locus of non-dominant poles and zeros, i.e. those to the left of $-\lambda$. For a passive $G(s)$, the non-dominant poles will converge to the non-dominant zeros, remaining to the left of $-\lambda$. 

\begin{figure}[htbp]
	\centering
	\includegraphics[width=0.35\textwidth]{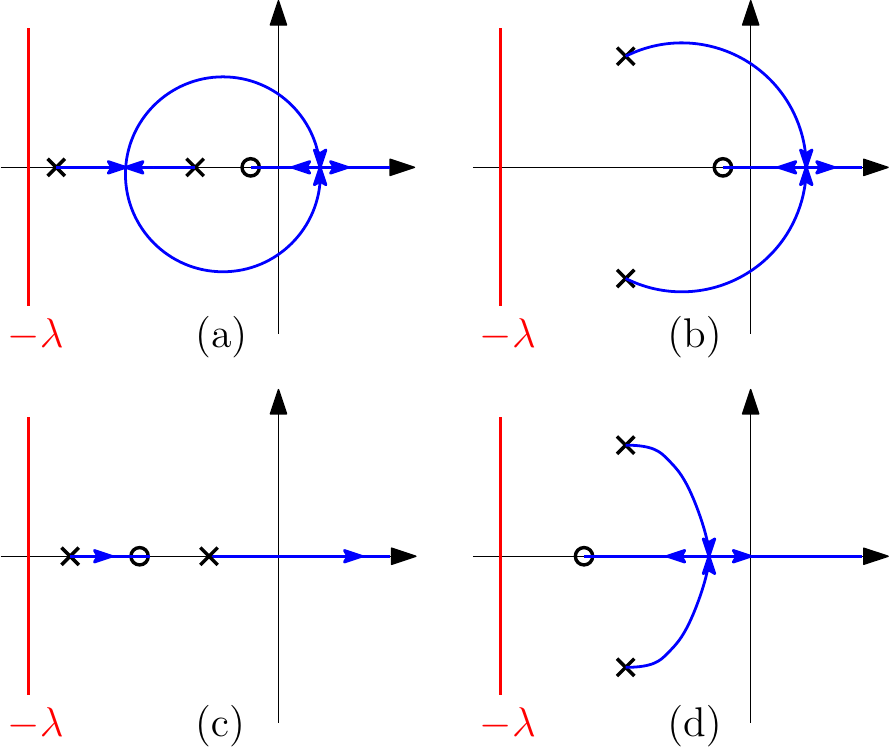}
	\caption{Positive root locus of $G(s)$ with $2$ dominant poles. (a),(b) $G(s)$ has a slow zero; (c),(d) $G(s)$ has no slow zero.}
	\label{fig:Rlocus_Zero}
\end{figure}

 Figure \ref{fig:Rlocus_Zero} (a),(b) are the root locus of the dominant poles of $G(s)$ with a slow zero. As $K$ increases, the two dominant poles will first become a pair of unstable complex conjugate and eventually will meet the right half real axis. One pole will keep moving to the right while the other will converge to the slow zero, crossing the imaginary axis at $K=\frac{1}{G(0)}$. By contrast in Figure \ref{fig:Rlocus_Zero} (c),(d), when there is no slow zero, only one dominant pole can cross the imaginary axis to the right half plane as $K$ increases. In these cases, Proposition \ref{prop:change_stable} shows that the origin is always stable for $K\leq\frac{1}{G(0)}$. 
Note that the zeros of $G(s)=\frac{C(s)}{1+2P(s)C(s)}$ are
given by the zeros of $C(s)$ and by the poles of $P(s)$. Thus, to introduce a slow zero in $G(s)$, we can either design $C(s)$ with a slow zero; or we can take advantage of any slow poles in $P(s)$.

\section{Oscillation control with RLC circuits}\label{sec:RLCcontroller}
Following Section \ref{sec:dominance} and \ref{sec:instability}, the design of the cross coupled controller proceeds as follows: 
\begin{enumerate}
    \item Given the plant admittance $P(s)$ and the XCP nonlinearity $\varphi$, select a dominant rate $\lambda$ and the impedance $C(s)$ such that
    \begin{itemize}
        \item $G(s)=\frac{C(s)}{1+2P(s)C(s)}$ has a slow zero;
        \item the conditions of Theorem \ref{th:inverse_circle_cirteria} are satisfied.
    \end{itemize}
    \item Induce the instability of the origin for $K=\sqrt{k_nI}\leq\frac{1}{G(0)}$, by tuning the current $I$.
\end{enumerate}

In this section we use the inverse circle criterion of Section \ref{sec:XCP_dominance} and the controlled instability of Section \ref{sec:instability} to derive conditions of oscillations for simple  impedances $C(s)$ given by the RLC and RC circuits of Figure \ref{fig:Controller_Circuit}. Under mild conditions, these simple circuits are enough to trigger oscillations.  

\begin{figure}[htbp]
	\centering
	\includegraphics[width=0.45\columnwidth]{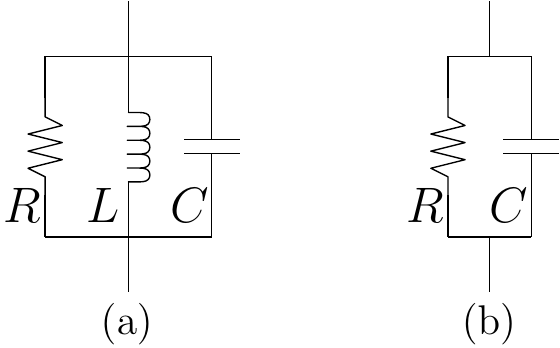}
	\caption{RLC and RC controller candidates for $C$.}
	\label{fig:Controller_Circuit}
\end{figure}

The two circuits have transfer function 
\begin{align}
    \begin{split}\label{eq:C_rlc}
         C_{rlc}(s)&=\frac{1}{1/R+1/Ls+Cs} =\frac{RLs}{RLCs^2+Ls+R}
    \end{split}
    \\
    C_{rc}(s)&=\frac{1}{1/R+Cs}\label{eq:C_rc}
\end{align}
where $R$, $L$ and $C$ are the resistance, inductance and capacitance of the circuit, respectively. The $RLC$ circuit has a slow zero at $0$, while the $RC$ circuit provides no zeros. In agreement with Section \ref{sec:instability}, the latter needs to be paired with 
plant admittances $P(s)$ which have a slow pole. 

The shifted inverse transfer functions read
\begin{align}
    C^{-1}_{rlc}(s-\lambda)&=\underbrace{\frac{1}{R}-C\lambda}_{\text{horizontal shift}}+\frac{1}{L(s-\lambda)}+Cs \,,\\
    C^{-1}_{rc}(s-\lambda)&=\underbrace{\frac{1}{R}-C\lambda}_{\text{horizontal shift}}+Cs\,.
\end{align}

\begin{figure}[htbp]
\vspace{5mm}
	\centering
	\includegraphics[width=0.92\columnwidth]{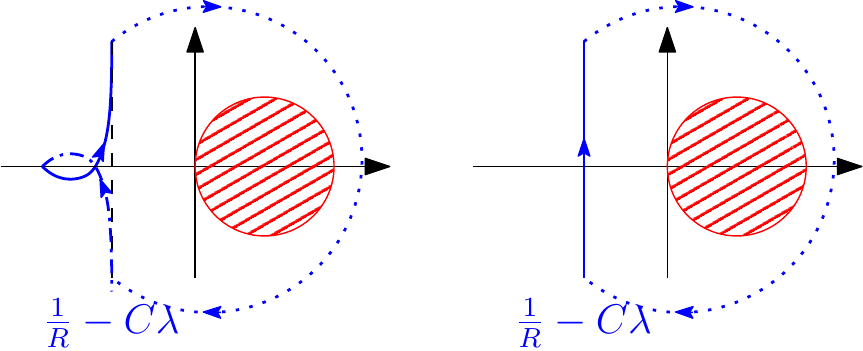}
	\caption{Nyquist plots of $C^{-1}_{rlc}(s-\lambda)$, lLeft, and $C^{-1}_{rc}(s-\lambda)$, Right, and the disk \eqref{eq:disk}.}
	\label{fig:Nplot_RLC_inver}
\end{figure}

Looking at their Nyquist plots in Figure \ref{fig:Nplot_RLC_inver}, both
Nyquist locii of $C^{-1}_{rlc}(s-\lambda)$ and $C^{-1}_{rc}(s-\lambda)$ 
encircle clockwise any
subregion to the right of the vertical line passing through $1/R-C\lambda$.
This feature leads to the following theorems.

\begin{theorem}\label{th:design_rlc}
    Consider the impedance $C(s) = C_{rlc}(s)$ in \eqref{eq:C_rlc} and a dominant rate $\lambda>0$. For any plant $P(s)$ that satisfies
    \begin{enumerate}
        \item $C_{rlc}^{-1}(s)+2P(s)$ has no zeros with real part at $-\lambda$,
        \item $P(s-\lambda)$ has no unstable poles,
        \item $\max_{\omega}(\Re(2P(j\omega-\lambda))+\frac{1}{R}-C\lambda< 0$,
    \end{enumerate}
    the closed-loop system in Figure \ref{fig:XCP_OSC_Block} has a stable limit cycle  
    for $K\leq\frac{1}{G(0)}$ sufficiently large. 
    
\end{theorem}

\begin{proof}
Condition 1) of Theorem \ref{th:design_rlc} corresponds to Condition $1)$ 
of Theorem \ref{th:inverse_circle_cirteria}. 

We now look at Condition 2) of Theorem \ref{th:inverse_circle_cirteria}. 
For $\lambda>0$, $C_{rlc}(s-\lambda)$ provides an unstable zero at $0$. That is $r=1$.
$P(s-\lambda)$ has no unstable poles, therefore $q=0$.
Thus Condition 2) of Theorem \ref{th:inverse_circle_cirteria} requires that the Nyquist plot of $C^{-1}_{rlc}(s-\lambda)+2P(s-\lambda)$ has $2-1=1$ clockwise encirclement of the origin. 
This is verified as follows: the Nyquist plot of $C^{-1}_{rlc}(s-\lambda)$ does one clockwise encirclement of the origin, as shown in Figure \ref{fig:Nplot_RLC_inver}, left.
Furthermore, Condition 3) of Theorem \ref{th:design_rlc} guarantees that
the Nyquist plot of $C^{-1}_{rlc}(s-\lambda)+2P(s-\lambda)$ has the same number of clockwise encirclements of the Nyquist plot of $C^{-1}_{rlc}(s-\lambda)$.
The latter follows from the fact that 
$
\Re ( C^{-1}_{rlc}(j \omega-\lambda)+2P(j \omega-\lambda) ) 
<
\Re ( 2P(j \omega-\lambda) ) + \frac{1}{R} - C\lambda < 0
$
for $\omega\in (-\infty,\infty)$.

The argument for Condition 3) of Theorem \ref{th:inverse_circle_cirteria} is similar. The Nyquist plot of $C^{-1}_{rlc}(s-\lambda)$ 
encircles the disk \eqref{eq:disk}. Therefore, 
$\Re ( C^{-1}_{rlc}(j \omega-\lambda)) + \frac{1}{R}-C\lambda < 0$ guarantees that
$C^{-1}_{rlc}(s-\lambda)+2P(s-\lambda)$ encircles the disk as well.

 Theorem \ref{th:inverse_circle_cirteria} thus guarantees that the closed loop
 is $2-$dominant. $G(s)$ has a slow zero at zero (from $C_{rlc}(s)$). 
 This means that the closed loop will have a single unstable equilibrium at zero
 for $K=\sqrt{k_nI}\leq\frac{1}{G(0)}$ sufficiently large. 
 By Theorem \ref{th:p-attractor}, the closed loop must have a stable limit cycle
 for $K=\sqrt{k_nI}\leq\frac{1}{G(0)}$ sufficiently large.
\end{proof}

\begin{theorem}\label{th:design_rc}
Consider the impedance $C(s) = C_{rc}(s)$ in \eqref{eq:C_rc} and a dominant rate $\lambda>0$. For any plant $P(s)$ that satisfies
    \begin{enumerate}
        \item $C_{rc}^{-1}(s)+2P(s)$ has no zeros with real part at $-\lambda$,
        \item $P(s-\lambda)$ has one unstable pole, which is to the right of all the poles of $C_{rc}(s-\lambda)$,
        \item $\max_{\omega}(\Re(2P(j\omega-\lambda))+\frac{1}{R}-C\lambda < 0$,
    \end{enumerate}
        the closed-loop system in Figure \ref{fig:XCP_OSC_Block} has a stable limit cycle  
    for $K\leq\frac{1}{G(0)}$ sufficiently large. 
\end{theorem}
\begin{proof}
The proof of Theorem \ref{th:design_rc} is very similar to the proof of Theorem  \ref{th:design_rlc}. The only difference is that $C_{rc}(s)$ has no zero therefore $r=0$. However, Condition 2) in Theorem \ref{th:design_rc} guarantees that  $q=1$.
It also guarantees that $G(s)$ has a slow zero from $P(s)$.
\end{proof}

\section{Example: Oscillation Control of a DC motor} \label{sec:example}
We close the paper with an example on a controlled DC motor, a common electro-mechanical actuator, represented by the block diagram in Figure \ref{fig:DC_motor_block}. The transfer function from the input voltage $V_m$ to the 
output current $I_m$ reads \begin{equation}
    P(s)=\frac{J_ms+b_m}{(L_ms+R_m)(J_ms+b_m)+k_m^2}.
\end{equation}
\begin{figure}[!h]
	\centering
	\includegraphics[width=0.4\textwidth]{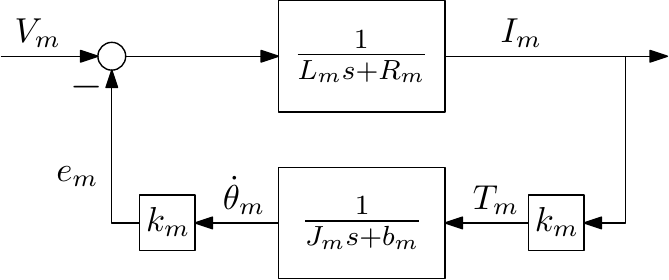}
	\caption{Block diagram representation of the DC motor.}
	\label{fig:DC_motor_block}
\end{figure}

To design our controller we consider nominal parameters $L_m=0.5$, $R_m=2$, $J_m=0.02$, $b_m=0.2$ and $k_m=0.1$. $P(s)$ has one zero at $-10$ and two poles at $-9.83$ and $-4.17$. 
The transistors in the XCP are characterized by $k_n=5$.

For slow oscillations we make $P(s)$ ``fast'' by selecting $\lambda=2$ such that $P(s-\lambda)$ has no right half plane poles. The following three RLC controller circuits all satisfy Theorem \ref{th:design_rlc}, as shown by the Nyquist plots in Figure \ref{fig:DC_motor_design_1} (Top Left): 
\begin{description}
    \item[\textbf{design 1:}] \hspace{.5cm} $R=100$, $L=1$, $C=1$;
    \item[\textbf{design 2:}] \hspace{.5cm} $R=100$, $L=1$, $C=5$;
    \item[\textbf{design 3:}] \hspace{.5cm} $R=100$, $L=5$, $C=1$.
\end{description}
For $I=2$, $K=\sqrt{k_nI}=3.16$, the origin of the controlled system is unstable for all three designs. This achieves stable oscillations with roughly the same magnitude, as shown by simulation results of the DC motor angular velocity $\dot{\theta}_m$ in Figure \ref{fig:DC_motor_design_1}. The resultant oscillation frequencies $\omega_i$ are close to the natural frequency of the RLC circuits $\sqrt{1/LC}$. 
For the same ratio $1/LC$, the choice of $C$ and $L$ affects the the shape of the oscillations, as shown in Design 2 and 3 of Figure \ref{fig:DC_motor_design_1}. The latter are closer to a relaxation-type oscillations and slower. 

For fast oscillation we take $\lambda=8$ and we make the pole of $P(s)$ at $-4.17$ the slowest dominant pole. The Nyquist plot of $P(s-\lambda)$ stays in the left half plane, i.e. $\max_{\omega}(\Re(2P(j\omega-\lambda))\leq 0$. Thus by Theorem \ref{th:design_rc}, every RC controller \eqref{eq:C_rc} with $4.17\leq\frac{1}{RC}\leq 8$ is a feasible controller for oscillation. Here we consider $R=1.5$, $C=0.1$, which leads to the Nyquist plot in Figure \ref{fig:DC_motor_design_2} Left. With this design choice, $G(0)=0.61$ and $I$ must be smaller than $1/k_nG(0)^2=0.54$ for a unique equilibrium point in the closed loop. For $I=0.5$, that is, $K=\sqrt{k_nI}=1.58$, we get stable oscillation as shown in Figure \ref{fig:DC_motor_design_2} Right. In comparison with Figure \ref{fig:DC_motor_design_1}, the small current $I$ restricts the magnitude of the oscillation. The dominant pole of $P(s)$ takes an active role for the instability and hence the oscillation. This suggests that the fast oscillation case is more sensitive to the variations of the motor parameters. 

\begin{figure}[h]
\vspace{-4mm}
	\centering
	\includegraphics[width=0.5\textwidth]{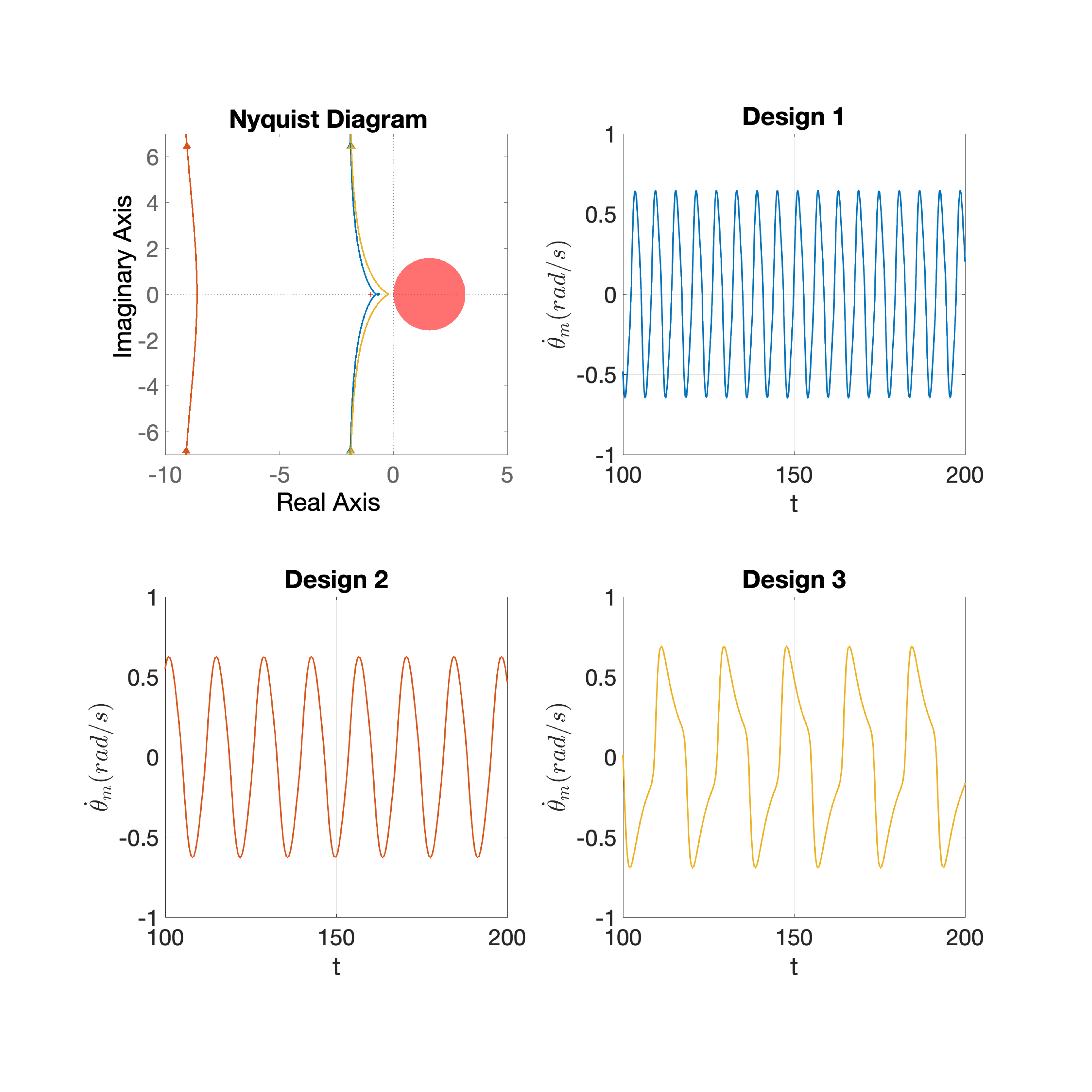}
	\vspace{-6mm}
	\caption{Nyquist plots of $C^{-1}_{rlc}(s-\lambda)+2P(s-\lambda)$ for $\lambda=2$,
	for the three different circuits. The disk \eqref{eq:disk} is for $K=3.16$.
	Simulation results (steady state): design 1 (blue) has frequency $\omega_1=1.07$ rad/s; design 2 (red) has frequency $\omega_2=0.52$ rad/s; design 3 (yellow) achieves $\omega_3=0.42$ rad/s.}
	\label{fig:DC_motor_design_1}
\end{figure}

\begin{figure}[h]
\vspace{-4mm}
	\centering
	\includegraphics[width=0.5\textwidth]{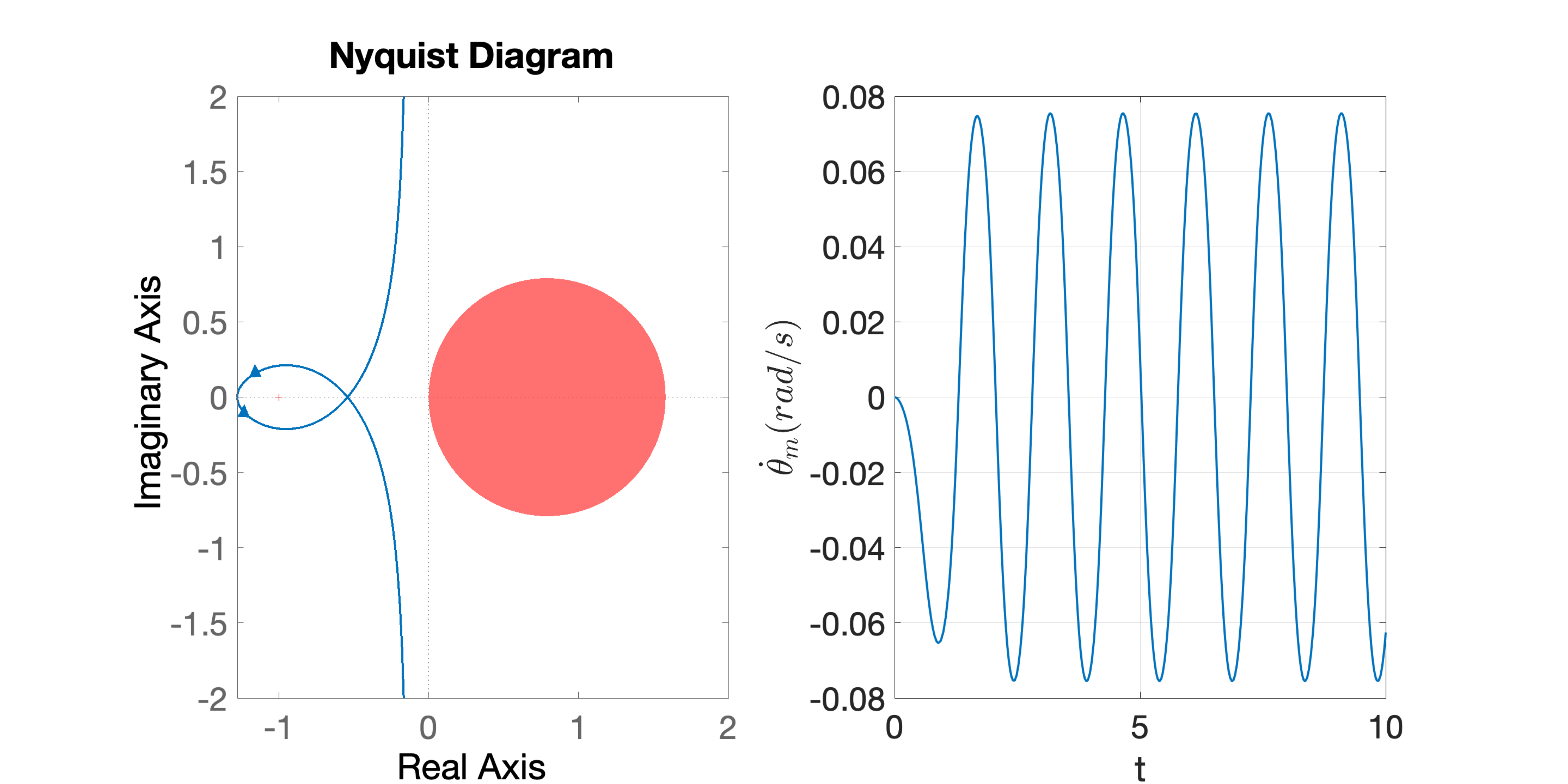}
	\caption{Left: Nyquist plot of $C^{-1}_{rc}(s-\lambda)+2P(s-\lambda)$ for $\lambda = 8$.
	The disk \eqref{eq:disk} is for $K=1.58$.  Right: the closed loop achieves the oscillation frequency $\omega=4.32$ rad/s.}
	\label{fig:DC_motor_design_2}
\end{figure}

\section{CONCLUSIONS}\label{sec:conclusiton} 
    We have shown that the cross coupled oscillator circuit can serve as an analog controller for oscillations. 
    Taking advantage of the XCP as a source of nonlinear positive feedback, we have used dominance theory to derive design conditions for oscillations. The inverse circle criterion for dominance provides guidance on the selection of circuit elements. Stable oscillations arise when the closed-loop system is $2$-dominant with all its equilibria unstable. Our design is also robust to uncertainties on the plant parameters, as small uncertainties induce small perturbations on the Nyquist locii, thus do not invalidate the conditions of the theorems (as long as perturbations are sufficiently small). The design is illustrated on a DC motor, a common electro-mechanical actuator.
   
    To use our approach in applications we will need a more systematic design for robustness and we will need to generalize
    the approach to nonlinear plants. This is particularly
    needed in robotics, to deal with nonlinear dynamics. The other
    important question is related to the features of the oscillation
    pattern. In this paper we have studied simple RLC and RC 
    circuits but larger circuits could be used to induce oscillations with specific features (harmonic vs relaxation, and specific patterns).

	\bibliographystyle{IEEEtran} 
	
\end{document}